\documentclass[11pt,a4paper]{article}
\usepackage[T1]{fontenc}
\usepackage{lmodern,amsmath,amssymb,amsthm,mathtools}
\usepackage[margin=25mm]{geometry}
\usepackage{graphicx,booktabs,tabularx,microtype,enumitem,xcolor}
\usepackage{placeins}
\usepackage[round,authoryear]{natbib}
\usepackage[colorlinks=true,allcolors=teal!60!black]{hyperref}
\hypersetup{pdftitle={Selective Inference for CART with Binary Outcomes},
 pdfsubject={Working manuscript on finite-sample conditional testing and power under parent homogeneity},
 pdfkeywords={selective inference, CART, binary outcome, Gini gain, conditional calibration}}
\newtheorem{theorem}{Theorem}[section]
\newtheorem{proposition}[theorem]{Proposition}

\theoremstyle{definition}
\theoremstyle{remark}
\newcommand{\PP}{\mathbb P}
\newcommand{\EE}{\mathbb E}
\newcommand{\1}{\mathbf 1}
\DeclareMathOperator{\Var}{Var}
\DeclareMathOperator{\Cov}{Cov}
\setlist{itemsep=3pt,topsep=5pt}
\allowdisplaybreaks
\title{Selective Inference for CART with Binary Outcomes}
\author{Tomoshige Nakamura \thanks{Faculty of Health Data Science, Juntendo University}}
\date{\today}
\begin{document}
\maketitle
\begin{abstract}
Binary classification trees select subgroups using the same outcomes later
used to assess their differences. We develop finite-sample conditional
tests of a common success probability within a parent selected by
deterministic Gini CART. The construction retains all eligible cutpoints
and conditions on the selected split, its ancestor path, the parent
success total, and outside outcomes. The resulting uniform label fiber
gives an exact count distribution, while a reversible parallel Monte
Carlo construction yields super-uniform inclusive and exactly uniform
tie-randomized p-values for any prespecified finite run budget. Within
a two-child constant-risk model, the selected count law is an exponential
family with information equal to its conditional count variance. We
separate information loss from computational limitations and exhibit a
selected fiber disconnected under single-label swaps. Simulations with
200 or 400 observations, ten independent or correlated predictors,
and trees of depth three show conservative inclusive tests and
nontrivial power for large risk differences. At 400 observations and
a generating risk difference of 0.4, randomized rejection conditional
on reaching the prespecified third-level target is 51--71\%, while
selection followed by rejection occurs in 11--21\% of datasets.
Smaller signals remain difficult to detect, and a representative
fivefold increase in computation gives little power improvement.
The guarantee concerns parent homogeneity, or equality of two constant
child risks, and does not cover equality of heterogeneous regional averages.
\end{abstract}

\section{Introduction}
Identifying groups with different outcome risks is a common objective
in data analysis. In medical and epidemiological studies, for example,
investigators may seek combinations of characteristics that distinguish
groups with different probabilities of an event. Classification and
regression trees (CART; \citealp{breiman1984}) provide an interpretable
approach to this task by recursively partitioning predictor space.
Each split produces two groups whose observed outcomes differ. When
groups are specified independently of the outcomes, their risks can be
compared using standard inferential methods under the relevant sampling
model. When CART selects the groups using those same outcomes, however,
the observed differences also reflect the search over candidate groups.
Inference must account for this selection to distinguish evidence of
a risk difference from variation favored by the search.

Selective inference addresses this problem by evaluating a comparison
conditional on the event that it was selected. Tree-Values
\citep{neufeld2022} develops this approach for regression trees, providing
tests for differences between selected sibling regions and confidence
intervals for regional means under a Gaussian response model. Binary
outcomes call for a conditional reference distribution that accounts
for Bernoulli sampling. Related work by \citet{strobl2007} derives the
exact null distribution of maximized Gini gain for a binary response
and uses it to construct a split-selection criterion. Here we retain
Gini maximization and test conditional on the winning split and the
ancestor history that produced its parent.

We study this question for unpruned Gini CART with prespecified growing
rules. We test the null hypothesis that all observations in the parent
of a selected split share a common success probability. In a model with one constant success
probability in each child, this null is equivalent to equality of the
two child risks. It restricts individual risks within the parent more
strongly than the regional-mean null considered by Tree-Values.
The target is local: independent outcome risks outside
the parent may differ, allowing a homogeneous region to be tested within
a tree that contains signal elsewhere. Our construction accommodates
multiple predictors, all eligible cutpoints, and selected internal nodes
at any finite depth by accounting for the ancestor splits that define
the target parent.

Our construction specializes classical fixed-total conditioning and the
selective exponential-family framework of \citet{fithian2014}.
Under the common-risk null, binary configurations with a given parent
success total have equal probability. Conditioning additionally on the
selected path retains only configurations reproducing the recorded
Gini comparisons. For a split below the root, we also fix outcomes
outside the parent. This gives a conditional distribution free of the
common risk and of the independent outside risks. The resulting
Fisher-type count comparison incorporates CART's search across
predictors and cutpoints.

Computing the conditional distribution by enumerating all admissible
configurations becomes expensive as the parent grows. We address this
computational problem using the parallel Monte Carlo construction of
\citet{besag1989}; see also \citet{howes2023}. We use reversible
label updates that preserve the target split and its ancestor path,
and construct tests from exchangeable reference endpoints. Both an
inclusive test and a test that randomizes within tied ranks have
finite-sample conditional guarantees for a prespecified computational
budget. The exchangeability guarantee follows from that existing
construction; the CART implementation specifies label swaps and checks
every winning comparison along the target path.

The contribution is a concrete conditional testing procedure for
multivariate Gini CART, including all eligible cutpoints, deterministic
tie resolution, and ancestor selection, together with an analysis of
its computational limitations and power. We exhibit a selected set of
label configurations that single swaps cannot connect. Standard
exponential-family arguments under the two-child model characterize
the conditional information and one-sided power available before
these computational restrictions are imposed.

Simulation studies with 200 or 400 observations
and ten independent or correlated predictors evaluate calibration at
depths up to three and investigate power at a third-level target.
We distinguish rejection conditional on reaching that target from
selection followed by rejection, and compare chain length with the
number of reference endpoints at fixed computational costs. Together, these
results describe the statistical and computational behavior of exact
selective testing for binary Gini CART under parent homogeneity.

Section~\ref{sec:setting} specifies the tree-growing rule, inferential
target, and selection event. Section~\ref{sec:homogeneity} derives the
conditional count distribution and tests. Section~\ref{sec:computation}
presents the Monte Carlo implementation and examines conditional
information and power. Section~\ref{sec:simulations} reports the
simulation results, and Section~\ref{sec:discussion} discusses their
implications and directions for further work.
\section{CART construction and the inferential target}
\label{sec:setting}
We consider an unpruned binary CART grown by maximizing Gini gain over
all eligible cutpoints. The growing rule and its deterministic treatment
of ties specify the selection event used for inference.

\subsection{Sampling, node regions, and depth}
\label{sec:tree-notation}
Fix a design $X=(X_1,\ldots,X_n)$ with fully observed predictors
$X_i\in\mathbb R^p$. Conditional on $X$, let
\[
 Y_i\overset{\mathrm{ind}}\sim\operatorname{Bernoulli}(p_i),
 \qquad i=1,\ldots,n.
\]
The success probabilities may vary across observations. All inferential
statements condition on $X$; no distributional assumption on the
predictors is required.

Index nodes by words $v\in\{L,R\}^{*}$, with the empty word
$\varnothing$ denoting the root. Write $R_v\subseteq\mathbb R^p$ for
the node region, $I_v=\{i:X_i\in R_v\}$ for its sample index set,
and $N_v=|I_v|$. The root has $R_{\varnothing}=\mathbb R^p$ and depth
zero. Splitting coordinate $j$ at threshold $t$ produces the regions
\begin{equation}
 R_{vL}=R_v\cap\{x:x_j\le t\},\qquad
 R_{vR}=R_v\cap\{x:x_j>t\}.
 \label{eq:recursive-regions}
\end{equation}
A coordinate may be used more than once along a path. Left and right
refer to predictor order, independently of the observed child risks.

Fix the maximum depth $D$ and a positive integer minimum child size
$n_{\min}$ before inspecting outcomes. A node at depth $|v|<D$ may
split; its split is reported as \emph{level} $|v|+1$. Thus $D=3$
permits tests at levels one, two, and three, with at most seven internal
nodes and eight leaves. Early stopping may yield fewer nodes.

\subsection{Eligible cutpoints}
\label{sec:candidate-families}
Within node $v$, sort feature $j$ as
$X^{v,j}_{(1)}\le\cdots\le X^{v,j}_{(N_v)}$. Rank $k$ is eligible if
\begin{equation}
 n_{\min}\le k\le N_v-n_{\min},\qquad
 X^{v,j}_{(k)}<X^{v,j}_{(k+1)}.
 \label{eq:eligible-rank}
\end{equation}
The corresponding threshold is
$t_{v,j,k}=\{X^{v,j}_{(k)}+X^{v,j}_{(k+1)}\}/2$. For candidate
$c=(j,k)$, denote its left and right sample index sets by $I_{v,c}^L$
and $I_{v,c}^R$, with sizes $a_c=k$ and $b_c=N_v-k$.
Equal predictor values are never separated.

The candidate family $\mathcal C_v(X)$ contains every eligible pair
$(j,k)$ across all $p$ features. With no predictor ties and
$N_v\ge2n_{\min}$, there are $p(N_v-2n_{\min}+1)$ candidates.
A constant feature supplies no eligible cut. Candidate cutpoints are
computed within each current parent; once that parent's membership is
fixed, they depend only on $X$.

\subsection{Gini maximization and the growing algorithm}
\label{sec:cart-algorithm}
For a nonempty sample index set $I$, write
$\overline Y_I=|I|^{-1}\sum_{i\in I}Y_i$ and define binary Gini
impurity $G(I)=2\overline Y_I(1-\overline Y_I)$. The impurity decrease
for candidate $c$ in node $v$ is
\begin{align}
 \operatorname{Gain}_v(c)
 &=G(I_v)-\frac{a_c}{N_v}G(I_{v,c}^L)
                -\frac{b_c}{N_v}G(I_{v,c}^R)\nonumber\\
 &=\frac{2a_cb_c}{N_v^2}
       (\overline Y_{I_{v,c}^L}-\overline Y_{I_{v,c}^R})^2
  =\frac{2}{N_v}T_{v,c}^2,\label{eq:gini-gain}\\
 T_{v,c}&=\sqrt{\frac{a_cb_c}{N_v}}
       (\overline Y_{I_{v,c}^L}-\overline Y_{I_{v,c}^R}).
 \label{eq:cart-score}
\end{align}
Thus the winner maximizes the squared scaled difference in observed
child risks. Its sign plays no role in selection.

Candidates have a fixed priority: smaller feature index first, then
smaller cut rank. Starting at the root, the fitted tree
$\widehat{\mathcal T}=\mathcal T(X,Y;D,n_{\min})$ is constructed as follows:
\begin{enumerate}[label=\arabic*.]
 \item If $|v|=D$ or $\mathcal C_v(X)$ is empty, declare $v$ terminal.
 \item Otherwise, compute all candidate gains and choose the first
 maximizer according to the fixed priority. If its gain is zero,
 declare $v$ terminal.
 \item If the maximum gain is positive, store its feature, rank, and
 midpoint, form the regions in \eqref{eq:recursive-regions}, and apply
 the same recursion to both children.
\end{enumerate}
All observations have equal weight. The tree is grown to the prescribed
depth subject to these local stopping rules and is not pruned.
In particular, $N_v<2n_{\min}$ or a pure parent forces termination.
Each node's decision depends only on observations in that node, so the
order of processing different branches does not affect the fitted tree.

For binary outcomes, put $M_v=\sum_{i\in I_v}Y_i$ and
$K_{v,c}=\sum_{i\in I_{v,c}^L}Y_i$. Then
\[
 T_{v,c}=\frac{N_vK_{v,c}-a_cM_v}{\sqrt{N_va_cb_c}}.
\]
Candidate scores can therefore be compared by cross-multiplying the
integer numerators and denominators of
$(N_vK_{v,c}-a_cM_v)^2/(a_cb_c)$. Our implementation uses these exact
comparisons within its documented integer ranges. Ties have positive
probability even with continuous predictors, making their deterministic
treatment part of the selection rule.

\subsection{The selected path and conditioning event}
\label{sec:selection-event}
Fix a possible target address $v\in\{L,R\}^{d}$ with $d<D$, and let
$v_0=\varnothing,v_1,\ldots,v_d=v$ be its successive prefixes. A path
history $h=(c_0,\ldots,c_d)$ records the selected feature and cut rank
at each ancestor and at the target split. Given $X$ and $h$, these
splits determine the parent memberships, thresholds, and competing
candidates along the path. Consider only histories whose splits are
eligible. Let $H_v(X,Y)$ denote the history selected by the recursion,
and set $H_v=\dagger$ if address $v$ is absent or terminal.

For a prescribed winner $c_a$ at node $v_a$, write $c\prec c_a$ when
competitor $c$ has higher tie priority. On the binary outcome space,
the event selecting that winner is
\begin{align}
 \mathcal E_{a,c_a}(X)=\{y\in\{0,1\}^n:\;&
 T_{v_a,c_a}(y)^2>0;\nonumber\\[-2pt]
 &T_{v_a,c_a}(y)^2>T_{v_a,c}(y)^2
                         \quad\text{for all }c\prec c_a;\nonumber\\[-2pt]
 &T_{v_a,c_a}(y)^2\ge T_{v_a,c}(y)^2
                         \quad\text{for all other }c\ne c_a\}.
 \label{eq:local-selection-event}
\end{align}
All competitors belong to $\mathcal C_{v_a}(X)$ for the parent
specified by $h$. The complete target-path event is
\begin{equation}
 E_v(h;X)=\{y:H_v(X,y)=h\}
         =\bigcap_{a=0}^{d}\mathcal E_{a,c_a}(X).
 \label{eq:path-selection-event}
\end{equation}
The equality follows recursively: the root winner fixes the next
parent, whose winning comparison then fixes the following parent.
Because the growing rules are local, splits in other branches do not
affect this event.

For example, testing address $v=LR$ at level three conditions on the
root winner, the winner in its left child, and the winner in that
child's right child. If the first two thresholds are $(j_0,t_0)$ and
$(j_1,t_1)$, the target parent is
$R_{LR}=\{x:x_{j_0}\le t_0,\ x_{j_1}>t_1\}$. Every feature and
eligible cut competes at each of these three nodes. Winners in the
root's right subtree and the $LL$ subtree are unrestricted.

For a fixed design $X$ and history $h$, let $A_h$ and $B_h$ be the
left- and right-child sample index sets at the target split, and let
$R_h=A_h\cup B_h$ be its parent index set. These sets are determined
by the recorded splits and $X$. They index observations, whereas
$R_v$ denotes a region in predictor space.

\subsection{The parent-homogeneity null and two-child model}
\label{sec:regional-targets}
Fix a possible history $h$ at the specified address $v$, with target
sets $R_h,A_h,B_h$ as defined in Section~\ref{sec:selection-event}.
We formulate hypotheses for these fixed sets as restrictions on the
success probabilities in the sampling model of
Section~\ref{sec:tree-notation}, and evaluate inference conditional on
$E_v(h;X)$. This construction applies separately to each possible
history; the observed CART history determines the target being tested.

The parent-homogeneity null is
\begin{equation}
 H_{0,h}^{\mathrm{hom}}:\quad
 \exists q\in(0,1)\text{ such that }p_i=q\quad(i\in R_h).
 \label{eq:strong-null}
\end{equation}
The probabilities outside $R_h$ may differ freely, subject to the
independence assumption in Section~\ref{sec:tree-notation}. A homogeneous
parent below the root can therefore occur in a tree with signal in
ancestors or other branches. The null is defined for this history's
target and is not imposed simultaneously on all candidate parents.
At $q=0$ or $1$, a positive-gain split within the parent has probability
zero. For a fixed $h$, we henceforth abbreviate $A_h,B_h,R_h$ as
$A,B,R$ and write $H_{0,R}^{\mathrm{hom}}$ for the same null.

To characterize the conditional alternative law and power, we also
consider the two-child model
\begin{equation}
 p_i=
 \begin{cases}
  p_A,&i\in A,\\
  p_B,&i\in B,
 \end{cases}
 \qquad 0<p_A,p_B<1.
 \label{eq:two-child-model}
\end{equation}
These restrictions concern the Bernoulli sampling law before conditioning
on the selected history. Within this model,
\eqref{eq:strong-null} is equivalent to $p_A=p_B$, and rejection
provides evidence of unequal child risks. The comparison concerns
immediate siblings of an internal node; either child may be split
further by the growing algorithm.

Section~\ref{sec:homogeneity} constructs the test conditional on $X$,
$E_v(h;X)$, the parent success total $M_R=\sum_{i\in R}Y_i$, and the
outside labels $Y_{R^c}$. The path event fixes the target split and its
ancestors without fixing their winning signs or subsequent splits.
Although outside labels are held fixed, every ancestor comparison
remains part of the selection check. The resulting guarantee concerns
a specified address and its selected path; it is nodewise rather than
simultaneous over the fitted tree.
\section{Exact conditional inference under parent homogeneity}
\label{sec:homogeneity}
To assess a selected split, we compare its observed child outcomes with
alternative binary outcome configurations that would have led CART to
select the same split and its ancestors. This comparison accounts for
the search that produced the groups. The task is to determine the null
distribution over these alternative configurations. We specialize
classical fixed-total conditioning and the selective exponential-family
framework of \citet{fithian2014} to the Gini CART path event in
Section~\ref{sec:selection-event}. The resulting law explicitly includes
competition across predictors and cutpoints at the target and its ancestors.

Under parent homogeneity, conditioning on the number of successes in
the parent provides a simple starting point: every arrangement of those
successes has the same probability, regardless of the unknown common
risk. Restricting these arrangements to those that reproduce the selected
path then gives the required conditional reference distribution. We first
derive this distribution, use it to construct two-sided tests, and then
describe its behavior under the two-child model with unequal risks.
Section~\ref{sec:computation} provides a Monte Carlo implementation when
direct enumeration is impractical.

\subsection{The conditional null distribution}
Fix the design $X$ and a feasible target-path history $h$ at address $v$,
as defined in Section~\ref{sec:selection-event}. Use the fixed target
sets $A=A_h$, $B=B_h$, and $R=R_h$ introduced there, and write
$E=E_v(h;X)$ for the selection event. In addition to $E$, we
condition on the parent success total $M_R=\sum_{i\in R}Y_i=M$ and
the outside labels $O=Y_{R^c}=o$. Denote this conditioning information
by $\mathcal D=(X,E,M_R=M,O=o)$ and assume its event has positive
probability. Holding $o$ fixed allows the probabilities outside the
parent to remain unspecified.

The admissible configurations within the parent form the set
\[
 \mathcal F=\{y_R\in\{0,1\}^{|R|}:\textstyle\sum_{i\in R}y_i=M,
                         \ (y_R,o)\in E\},
\]
where $(y_R,o)$ denotes the full outcome vector with these labels in
their original index positions. We call $\mathcal F$ the selected
\emph{fiber}. It retains exactly the fixed-total configurations that
satisfy every winning comparison along the target path.

\begin{samepage}
\begin{theorem}[Exact conditional law and selective validity]
\label{thm:exact}
Under $H_{0,R}^{\mathrm{hom}}$, conditional on $\mathcal D$, the
outcomes in $R$ are uniform on $\mathcal F$:
\[
 \PP_0(Y_R=y_R\mid\mathcal D)=\frac{1}{|\mathcal F|},
 \qquad y_R\in\mathcal F.
\]
For any statistic $V$ on this fiber, its inclusive upper-tail probability
is conditionally super-uniform: its rejection probability at threshold
$\alpha$ is at most $\alpha$ for every $\alpha\in[0,1]$.
Independent uniform randomization within equal values of $V$ gives a
conditionally uniform p-value.
\end{theorem}
\end{samepage}
\begin{proof}
Under the null, each parent configuration with $M$ successes has
probability $q^M(1-q)^{|R|-M}$. Independence makes this factor unchanged
by conditioning on the outside labels. Restricting to the nonempty
set $\mathcal F$ and normalizing cancels the common factor, proving
uniformity. Ordering the possible values of $V$ gives the inclusive
tail bound. Spreading each probability atom uniformly over its
corresponding tail-probability interval gives the randomized result.
\end{proof}

We use the left-child success count $K=\sum_{i\in A}Y_i$ as the test
statistic. With $a=|A|$, $b=|B|$, and $M$ fixed, this count determines
the observed risk difference:
\[
 \overline Y_A-\overline Y_B=\frac{K}{a}-\frac{M-K}{b}.
\]
Let $C_{\mathcal D}(k)=\#\{y_R\in\mathcal F:K(y_R)=k\}$ count the
admissible configurations with left-child count $k$. Theorem~\ref{thm:exact}
yields
\begin{equation}
 \PP_0(K=k\mid\mathcal D)
 =\frac{C_{\mathcal D}(k)}{\sum_j C_{\mathcal D}(j)}.
 \label{eq:strong-count-law}
\end{equation}
Without the selection restriction, these counts are
$\binom{a}{k}\binom{b}{M-k}$, giving the hypergeometric reference law
used in Fisher's exact test. CART selection replaces them with the
counts of configurations reproducing the target path. The resulting
reference law therefore incorporates competition across features and
cutpoints at every ancestor as well as at the tested split.

\subsection{Two-sided selective tests}
For an observed count $k$, an inclusive, equal-tailed two-sided p-value is
\begin{equation}
 p_{\mathrm{inc}}(k)=\min\{1,\,
  2\min[\PP_0(K\le k\mid\mathcal D),\PP_0(K\ge k\mid\mathcal D)]\}.
 \label{eq:inclusive-exact}
\end{equation}
Here \emph{inclusive} means that the entire probability mass at the
observed count is included in both tails. A randomized counterpart uses
\begin{equation}
 r_0(k,U)=\PP_0(K>k\mid\mathcal D)
              +U\PP_0(K=k\mid\mathcal D),\qquad
 p_{\mathrm{rand},0}=2\min\{r_0,1-r_0\},
 \label{eq:randomized-exact}
\end{equation}
where $U\sim\operatorname{Uniform}(0,1)$ is independent of the outcomes.
Randomization spreads the mass at the observed count over its rank
interval. These procedures assess both directions of the child risk
difference using the same conditional count law.

For either procedure, reject when $p\le\alpha$. Each inclusive tail
is super-uniform by Theorem~\ref{thm:exact}, so a union bound gives
conditional rejection probability at most $\alpha$ for
\eqref{eq:inclusive-exact}. For \eqref{eq:randomized-exact}, $r_0$
is uniform under the null, and therefore so is $2\min\{r_0,1-r_0\}$.
The randomized test has conditional rejection probability exactly
$\alpha$. Averaging over the parent total and outside labels preserves
these respective guarantees conditional on $X$ and $E$.

\subsection{The conditional law under unequal child risks}
The two-child model in \eqref{eq:two-child-model} also gives a tractable
conditional distribution under alternatives. Define the log odds ratio
$\theta=\operatorname{logit}(p_A)-\operatorname{logit}(p_B)$, with
$\theta=0$ corresponding to the null. Keep the same conditioning
information $\mathcal D$ and configuration counts $C_{\mathcal D}(k)$.

\begin{proposition}[Conditional count law under the two-child model]
Under model~\eqref{eq:two-child-model},
\[
 \PP_\theta(K=k\mid\mathcal D)=
 \frac{C_{\mathcal D}(k)e^{\theta k}}
      {\sum_j C_{\mathcal D}(j)e^{\theta j}}.
\]
\end{proposition}
\begin{proof}
With $M_R=M$, a label configuration with $K=k$ has inside likelihood
\[
 (1-p_A)^a(1-p_B)^b
 \left\{\frac{p_B}{1-p_B}\right\}^{M}e^{\theta k}.
\]
The factors preceding $e^{\theta k}$ are common on the fixed-total space.
Condition on outside labels, restrict to $E$, sum over the
$C_{\mathcal D}(k)$ configurations, and normalize.
\end{proof}
The selected count law is thus a one-parameter exponential family.
Selection determines the weights $C_{\mathcal D}(k)$, while the two
child risks enter the conditional distribution only through $\theta$.
This representation underlies the conditional
information and power results in Section~\ref{sec:information}.
\section{Finite-run computation and conditional information}\label{sec:computation}
Section~\ref{sec:homogeneity} gives an exact conditional reference
distribution, but enumerating all admissible outcome configurations can
be computationally expensive. Section~\ref{sec:mc-procedures} develops Monte Carlo tests
that use the same null model and conditioning event while retaining
finite-sample validity for a prespecified computational budget. We
construct reversible label updates and arrange the simulation runs so
that the observed and reference outcomes are exchangeable under the
conditional null. This construction yields inclusive and tie-randomized
tests. We establish their validity and describe how their rejection
rates are summarized in the simulations.
Finite-run validity follows from exchangeability under the conditional
null and does not require the chain to mix over the entire selected
fiber. This guarantee is distinct from numerical convergence to the
full conditional tail probabilities in Section~\ref{sec:homogeneity}.
In particular, increasing the number of reference endpoints at a fixed
chain length need not recover those probabilities.

Section~\ref{sec:information} examines how conditioning and computation affect power.
Within the two-child model, we characterize the information retained
by the conditional count distribution and the power of one-sided tests.
We then study how the label updates explore the admissible
configurations, including a case in which single swaps cannot connect
the entire selected fiber. Together, these results provide a basis for
interpreting the power and computational diagnostics reported in
Section~\ref{sec:simulations}.

\subsection{Construction of the Monte Carlo tests}
\label{sec:mc-procedures}
We apply the parallel Monte Carlo method of \citet{besag1989}, using
the exchangeability formulation described by \citet{howes2023}.
The method supplies the finite-run validity argument; here we specify
a reversible label-swap kernel that preserves the full Gini CART path
event and derive the inclusive and randomized count tests it produces.

Throughout this section, fix a positive-probability conditioning event
$\mathcal D$ from Section~\ref{sec:homogeneity}. Put $m=|R|$ and
$\pi_0(y)=|\mathcal F|^{-1}$ for $y\in\mathcal F$. The selected
positive-gain split ensures $1\le M\le m-1$. A transition proposes
exchanging one uniformly chosen success and one uniformly chosen failure
within $R$. The exchange is accepted if the resulting full outcome vector,
with outside labels $o$ restored, still satisfies $E$; otherwise the state
is unchanged. Thus, writing $d_H$ for Hamming distance, the transition
matrix on $\mathcal F$ is
\begin{equation}
 Q(y,z)=
 \begin{cases}
  \{M(m-M)\}^{-1}\1\{d_H(y,z)=2\},&z\ne y,\\[2pt]
  1-\displaystyle\sum_{w\in\mathcal F\setminus\{y\}}Q(y,w),&z=y.
 \end{cases}
 \label{eq:swap-kernel}
\end{equation}
A pair of fixed-total configurations at Hamming distance two determines
exactly one success/failure exchange. The off-diagonal entries are
therefore symmetric, and
\begin{equation}
 \pi_0(y)Q(y,z)=\pi_0(z)Q(z,y),\qquad y,z\in\mathcal F.
 \label{eq:swap-balance}
\end{equation}
Every proposal, including a rejected one, counts as a transition.
All target and ancestor comparisons are checked when evaluating $E$.

Choose integers $B,L\ge1$ and the transition rule as functions of
$\mathcal D$ only: for a fixed $\mathcal D$, the same $Q,B,L$ must
apply to every possible observed configuration in $\mathcal F$.
In the experiments, $B$ and $L$ are fixed constants. Write $P=Q^L$.
Starting from the observed parent labels $Y_0$, take $L$ transitions
to a common intermediate state $H$. From $H$, generate $B$ further
chains of length $L$, using independent random draws for each chain
and independently of the initial chain. Their endpoints are
$Y_1,\ldots,Y_B$. Conditional on $\mathcal D$, the sampling rule is
\begin{equation}
 H\mid Y_0=y_0\sim P(y_0,\cdot),\qquad
 \PP(Y_1=y_1,\ldots,Y_B=y_B\mid H=z,Y_0=y_0)
       =\prod_{b=1}^{B}P(z,y_b).
 \label{eq:hub-sampling-rule}
\end{equation}
Each $Y_b$ denotes a vector of parent labels; the outside labels stay
fixed. Use these same endpoints for both procedures below. Define
$K_b=K(Y_b)$, $N=B+1$, and
\[
 g=\sum_{b=1}^B\1\{K_b>K_0\},\qquad
 e=1+\sum_{b=1}^B\1\{K_b=K_0\}.
\]
Thus $g$ counts endpoints strictly above the observed count, while $e$
counts its entire tied block, including the observation itself.
In particular, $e\ge1$.

The inclusive procedure includes all endpoints tied with $K_0$ in each
tail, including $K_0$ itself:
\begin{equation}
 p_+=\frac{1+\sum_{b=1}^B\1\{K_b\ge K_0\}}{N}=\frac{g+e}{N},\qquad
 p_-=\frac{1+\sum_{b=1}^B\1\{K_b\le K_0\}}{N}=\frac{N-g}{N}.
 \label{eq:inclusive-mc-tails}
\end{equation}
Report $p_{\mathrm{MC,inc}}=\min\{1,2\min(p_+,p_-)\}$ and reject when
$p_{\mathrm{MC,inc}}\le\alpha$. No additional tie-breaking random number
is used: the decision is fixed given the observed data and the endpoint
sample. This procedure still contains Monte Carlo randomness from generating
those endpoints. Including the whole tied block makes the test conservative,
potentially substantially so for discrete outcomes.

For the randomized procedure, draw one
$U\sim\operatorname{Uniform}(0,1)$ independently of the data and all
Monte Carlo chains after generating the reference endpoints. Set
\begin{equation}
 r=\frac{g+Ue}{N},\qquad
 p_{\mathrm{rand}}=2\min(r,1-r),
 \label{eq:randomized-mc}
\end{equation}
and reject when $p_{\mathrm{rand}}\le\alpha$. Conditional on the endpoint
counts, $r$ is uniform on the rank interval $[g/N,(g+e)/N]$ occupied by
the observed tie block. The same $U$ determines the two complementary tails;
the tails are not randomized separately. This is randomization of the final
test, not of CART's split selection or the observed outcomes. A single
draw is used for each reported test and must not be redrawn to obtain a
preferred conclusion. Replacing $U$ by $1/2$ would give a mid-rank procedure,
which does not have the exact-uniform guarantee below.

\begin{proposition}[Finite-run conditional validity]\label{prop:finite-mc}
Assume $H_{0,R}^{\mathrm{hom}}$ and the sampling construction
\eqref{eq:swap-kernel}--\eqref{eq:hub-sampling-rule}, with $Q,B,L$
fixed conditional on $\mathcal D$ and the final $U$ independent of
all data and chains. For every $\alpha\in[0,1]$,
\begin{equation}
 \PP_0(p_{\mathrm{rand}}\le\alpha\mid\mathcal D)=\alpha,\qquad
 \PP_0(p_{\mathrm{MC,inc}}\le\alpha\mid\mathcal D)\le\alpha.
 \label{eq:finite-mc-validity}
\end{equation}
The conclusion also holds for any fixed transition matrix on
$\mathcal F$ reversible with respect to $\pi_0$, using the same
construction. Irreducibility, aperiodicity, and a mixing-time bound
are not required.
\end{proposition}
\begin{proof}
All probabilities in the proof condition on $\mathcal D$. Reversibility
of $Q$ implies reversibility of $P=Q^L$. Indeed, with
$D_\pi=\operatorname{diag}\{\pi_0(y):y\in\mathcal F\}$,
\[
 D_\pi Q=Q^\top D_\pi
 \quad\Longrightarrow\quad
 D_\pi Q^L=(Q^\top)^L D_\pi
 \quad\Longrightarrow\quad
 \pi_0(y)P(y,z)=\pi_0(z)P(z,y).
\]
By Theorem~\ref{thm:exact}, $Y_0$ has law $\pi_0$. The joint mass of
the observed state, the intermediate state, and all endpoints is therefore
\begin{align}
 &\PP_0(Y_0=y_0,H=z,Y_1=y_1,\ldots,Y_B=y_B\mid\mathcal D)\nonumber\\
 &\qquad=\pi_0(y_0)P(y_0,z)\prod_{b=1}^{B}P(z,y_b)
         =\pi_0(z)\prod_{b=0}^{B}P(z,y_b).
 \label{eq:hub-factorization}
\end{align}
Summing over $y_0,\ldots,y_B$ gives
$\PP_0(H=z\mid\mathcal D)=\pi_0(z)>0$. Consequently,
\begin{equation}
 \PP_0(Y_0=y_0,\ldots,Y_B=y_B\mid H=z,\mathcal D)
       =\prod_{b=0}^{B}P(z,y_b).
 \label{eq:conditional-iid-endpoints}
\end{equation}
Thus the $B+1$ endpoints, including the observation, are independent
and identically distributed conditional on $(H,\mathcal D)$.

Let $\mathcal M$ be the unordered multiset of the $N$ counts. For a
realized multiset, let $t_1<\cdots<t_J$ be its distinct values, with
multiplicities $n_1,\ldots,n_J$, and put $s_j=\sum_{\ell>j}n_\ell$.
Exchangeability gives
\[
 \PP_0(K_0=t_j\mid H,\mathcal D,\mathcal M)=\frac{n_j}{N}.
\]
On $\{K_0=t_j\}$, $(g,e)=(s_j,n_j)$, so the independent $U$
makes $r$ uniform on
$I_j=[s_j/N,(s_j+n_j)/N]$. These intervals partition $[0,1]$
up to their endpoints. For every Borel set $A\subseteq[0,1]$,
writing $\lambda$ for Lebesgue measure,
\begin{align*}
 \PP_0(r\in A\mid H,\mathcal D,\mathcal M)
 &=\sum_{j=1}^{J}\frac{n_j}{N}
       \frac{\lambda(A\cap I_j)}{n_j/N}
  =\lambda(A).
\end{align*}
Averaging over $(H,\mathcal M)$ proves $r\mid\mathcal D\sim
\operatorname{Uniform}(0,1)$. Hence
\[
 \PP_0\{2\min(r,1-r)\le\alpha\mid\mathcal D\}
 =\lambda\bigl([0,\alpha/2]\cup[1-\alpha/2,1]\bigr)=\alpha.
\]
Finally, using the same $U$ to compare the two procedures, pointwise
\[
 p_+-r=\frac{(1-U)e}{N}\ge0,\qquad
 p_--(1-r)=\frac{Ue}{N}\ge0.
\]
Since $2\min(r,1-r)\le1$, these inequalities imply
$p_{\mathrm{MC,inc}}\ge p_{\mathrm{rand}}$. The inclusion
$\{p_{\mathrm{MC,inc}}\le\alpha\}\subseteq
\{p_{\mathrm{rand}}\le\alpha\}$ proves the second part of
\eqref{eq:finite-mc-validity}.
\end{proof}
The finite-run guarantee uses the joint distribution of the observation and the simulated
endpoints. The conditional independence in
\eqref{eq:conditional-iid-endpoints} follows from the common
intermediate-state construction. Our swap implementation checks Gini
selection using integer comparisons within its documented input range.

In the simulation tables, the \emph{inclusive} columns average the indicators
$\1\{p_{\mathrm{MC,inc}}\le\alpha\}$. The \emph{randomized} columns
integrate out only the final uniform $U$ analytically, to reduce Monte Carlo
noise. Specifically, let $\ell=g/N$, $h=(g+e)/N$, and
$[x]_+=\max(x,0)$. Each tested target contributes
\begin{equation}
 q_\alpha(g,e)
 =\PP_U(p_{\mathrm{rand}}\le\alpha\mid K_0,\ldots,K_B)
 =\frac{[\min(h,\alpha/2)-\ell]_+
       +[h-\max(\ell,1-\alpha/2)]_+}{h-\ell}.
 \label{eq:integrated-randomized-rejection}
\end{equation}
This is the fraction of the rank interval lying in the two rejection
tails. The reported randomized rate is the average of $q_\alpha(g,e)$
over the eligible tests, not the proportion for which $q_\alpha(g,e)$
is at most $\alpha$: $q_\alpha$ is a rejection probability, not a p-value.
Its average estimates the rejection rate of the single-$U$ procedure in
\eqref{eq:randomized-mc}. Data generation and endpoint sampling remain
random. In Supplement~\ref{sec:strong-under-weak}, multiple nodes from one
training dataset are handled using dataset-clustered Monte Carlo standard
errors.

\subsection{Computation and statistical power}
\label{sec:information}
The conditional two-child model provides an analytical benchmark for
assessing the power of the Monte Carlo calculation. The information
identity and one-sided power ordering below are standard consequences
of its one-parameter exponential-family form. Set
$\mathcal K=\{k:C_{\mathcal D}(k)>0\}$ and
$Z_{\mathcal D}(\theta)=\sum_{k\in\mathcal K}C_{\mathcal D}(k)e^{\theta k}$.
All quantities in the following statements fix the same $\mathcal D$.

\begin{proposition}[Conditional information and a degenerate fiber]
\label{prop:information}
Fix $\mathcal D$ and assume the two-child model. For each finite
$\theta\in\mathbb R$, the conditional score and Fisher information are
\begin{equation}
 \partial_\theta\log\PP_\theta(K\mid\mathcal D)
      =K-\EE_\theta(K\mid\mathcal D),\qquad
 I_{\mathcal D}(\theta)=\Var_\theta(K\mid\mathcal D).
 \label{eq:conditional-information}
\end{equation}
The information is zero if and only if $\mathcal K$ is a singleton.
In that case, the full conditional distribution of $Y_R$ is independent
of $\theta$, and every conditional level-$\alpha$ test has power at
most $\alpha$ for every $\theta$.
\end{proposition}
\begin{proof}
The conditional masses of an individual configuration and of its count
are, respectively,
\begin{equation}
 \pi_\theta(y)=\frac{e^{\theta K(y)}}{Z_{\mathcal D}(\theta)},
 \quad y\in\mathcal F,\qquad
 f_\theta(k)=\frac{C_{\mathcal D}(k)e^{\theta k}}
                   {Z_{\mathcal D}(\theta)},\quad k\in\mathcal K.
 \label{eq:conditional-label-and-count-laws}
\end{equation}
The set $\mathcal K$ is finite and fixed as $\theta$ varies, and all
its masses are positive. For $j=1,2$,
\[
 Z_{\mathcal D}^{(j)}(\theta)
   =\sum_{k\in\mathcal K}k^j C_{\mathcal D}(k)e^{\theta k}.
\]
Writing $\psi_{\mathcal D}(\theta)=\log Z_{\mathcal D}(\theta)$ gives
\begin{align*}
 \psi_{\mathcal D}'(\theta)
  &=\frac{Z_{\mathcal D}'(\theta)}{Z_{\mathcal D}(\theta)}
    =\sum_k k f_\theta(k)=:\mu_\theta,\\
 \psi_{\mathcal D}''(\theta)
  &=\frac{Z_{\mathcal D}''(\theta)}{Z_{\mathcal D}(\theta)}
        -\left\{\frac{Z_{\mathcal D}'(\theta)}
                      {Z_{\mathcal D}(\theta)}\right\}^2
    =\sum_k(k-\mu_\theta)^2 f_\theta(k).
\end{align*}
Differentiating $\log f_\theta(k)=\log C_{\mathcal D}(k)
+\theta k-\psi_{\mathcal D}(\theta)$ yields the score $k-\mu_\theta$.
Its expected square is $\psi_{\mathcal D}''(\theta)$, proving
\eqref{eq:conditional-information}. Because every $f_\theta(k)>0$,
this variance vanishes exactly when $\mathcal K$ has one element.

If $\mathcal K=\{k_*\}$, then
$Z_{\mathcal D}(\theta)=|\mathcal F|e^{\theta k_*}$ and
$\pi_\theta(y)=1/|\mathcal F|$ for every $y\in\mathcal F$.
Represent any randomized test by its rejection probability
$\varphi(y)\in[0,1]$ given the observed configuration, averaging over
its auxiliary randomness. For a fixed test rule, $\varphi$ does not
depend on the unknown $\theta$. If it has conditional level $\alpha$,
\[
 \EE_\theta\{\varphi(Y_R)\mid\mathcal D\}
  =\frac{1}{|\mathcal F|}\sum_{y\in\mathcal F}\varphi(y)
  =\EE_0\{\varphi(Y_R)\mid\mathcal D\}\le\alpha
 \qquad(\theta\in\mathbb R).
\]
This also covers tests whose auxiliary randomness generates Monte Carlo
reference outcomes.
\end{proof}

\begin{samepage}
\begin{proposition}[One-sided conditional power ordering]
\label{prop:one-sided}
Fix $\mathcal D$, assume the two-child model with
$|\mathcal K|\ge2$, and let $0<\alpha<1$.
Choose $c\in\mathcal K$ satisfying
\[
 \PP_0(K>c\mid\mathcal D)\le\alpha
       \le\PP_0(K\ge c\mid\mathcal D),
\]
and define
\begin{equation}
 \gamma=\frac{\alpha-\PP_0(K>c\mid\mathcal D)}
                 {\PP_0(K=c\mid\mathcal D)},\qquad
 \varphi_\alpha(k)=\1\{k>c\}+\gamma\1\{k=c\}.
 \label{eq:one-sided-test}
\end{equation}
The test rejects with probability $\varphi_\alpha(K)$, using independent
boundary randomization. It is uniformly most powerful among conditional
level-$\alpha$ tests for $\theta=0$ against $\theta>0$. Its power is
nondecreasing in $\theta$, and it is also a level-$\alpha$ test of
$\theta\le0$. The corresponding lower-tail test has the analogous
properties against $\theta<0$.
\end{proposition}
\end{samepage}
\begin{proof}
All expectations below condition on the fixed $\mathcal D$. Since the
null distribution has finite support with positive masses, such a $c$
exists, its denominator in \eqref{eq:one-sided-test} is positive, and
$0\le\gamma\le1$. In particular,
\[
 \EE_0\varphi_\alpha(K)=\PP_0(K>c)+\gamma\PP_0(K=c)=\alpha.
\]
For any fixed $\theta_1>0$, the likelihood ratio on the full set of
configurations is
\[
 \Lambda_{\theta_1}(y)=\frac{\pi_{\theta_1}(y)}{\pi_0(y)}
   =\frac{Z_{\mathcal D}(0)}{Z_{\mathcal D}(\theta_1)}
                      e^{\theta_1 K(y)}.
\]
Set $\lambda_c=Z_{\mathcal D}(0)e^{\theta_1 c}/
Z_{\mathcal D}(\theta_1)$. For any competing test
$\eta:\mathcal F\to[0,1]$ with $\EE_0\eta\le\alpha$,
\[
 \{\eta(y)-\varphi_\alpha(K(y))\}
          \{\Lambda_{\theta_1}(y)-\lambda_c\}\le0
 \quad\text{for every }y\in\mathcal F.
\]
Indeed, the first factor is nonpositive when $K(y)>c$, nonnegative
when $K(y)<c$, and the second factor is zero when $K(y)=c$.
Consequently,
\begin{align*}
 \EE_{\theta_1}\eta-\EE_{\theta_1}\varphi_\alpha(K)
 &=\EE_0\bigl[(\eta-\varphi_\alpha(K))
                     (\Lambda_{\theta_1}-\lambda_c)\bigr]
       +\lambda_c\{\EE_0\eta-\alpha\}\le0.
\end{align*}
The same $c$ and $\gamma$ apply to every $\theta_1>0$, which proves
uniform most powerfulness, including among tests using the full labels.

For $\beta(\theta)=\sum_k\varphi_\alpha(k)f_\theta(k)$, the score
identity gives $f_\theta'(k)=(k-\mu_\theta)f_\theta(k)$. Differentiating
the finite sum and symmetrizing yields
\begin{align*}
 \beta'(\theta)
 &=\Cov_\theta\{K,\varphi_\alpha(K)\}\\
 &=\frac12\sum_{k,\ell\in\mathcal K}f_\theta(k)f_\theta(\ell)
       (k-\ell)\{\varphi_\alpha(k)-\varphi_\alpha(\ell)\}\ge0,
\end{align*}
since $\varphi_\alpha$ is nondecreasing. Thus
$\sup_{\theta\le0}\beta(\theta)=\beta(0)=\alpha$.
Replacing $(K,\theta)$ by $(-K,-\theta)$ proves the lower-tail claim.
\end{proof}

The optimality and power ordering concern a prespecified direction in
the fixed conditional two-child experiment. They do not extend to the
two-sided Monte Carlo test or to comparisons in which the selected tree
changes with signal strength. General departures from homogeneity need
not have a scalar $\theta$.

For the Monte Carlo tests, if $K_0=K_1=\cdots=K_B$, then $(g,e)=(0,N)$,
$p_{\mathrm{MC,inc}}=1$, and $p_{\mathrm{rand}}=2\min(U,1-U)$, giving
$q_\alpha=\alpha$. An identity kernel produces this behavior under
every alternative. Randomization resolves rank discreteness but does
not recover information lost through conditioning or poor exploration.
A randomized rejection rate near 5\% therefore needs to be interpreted
together with power and endpoint-movement diagnostics.

All-tied endpoints do not imply that $\mathcal K$ is a singleton:
other configurations may be inaccessible or not reached within the
run budget. The next result shows that even arbitrarily long single-swap
chains need not explore the whole selected fiber.

\begin{proposition}[A disconnected selected-root fiber]
\label{prop:disconnected}
Let $n$ be divisible by four. Fix the design and a selected root split
$c_*$ with child sizes $a=b=n/2$, and condition on $M=n/2$.
Let $E$ specify this positive-gain root winner with the deterministic
candidate and tie rules of Section~\ref{sec:setting}, without its sign.
If $\mathcal F$ is nonempty, the single-swap matrix
\eqref{eq:swap-kernel} is reducible on $\mathcal F$. Moreover, for
every $y\in\mathcal F$ and integer $\ell\ge0$,
\begin{equation}
 \|Q^\ell(y,\cdot)-\pi_0\|_{\mathrm{TV}}\ge\frac12,
 \label{eq:disconnected-tv}
\end{equation}
where $\|\mu-\nu\|_{\mathrm{TV}}=
\sup_{S\subseteq\mathcal F}|\mu(S)-\nu(S)|$.
\end{proposition}
\begin{proof}
Write $r_*=n/4\in\mathbb Z$. For the prescribed split,
\[
 T_{c_*}(y)=\frac{nK(y)-(n/2)(n/2)}{\sqrt{n(n/2)(n/2)}}
           =\frac{2}{\sqrt n}\{K(y)-r_*\}.
\]
Positive selected gain excludes $K(y)=r_*$. Thus $\mathcal F$ is
the disjoint union of
\[
 \mathcal F_- =\{y\in\mathcal F:K(y)<r_*\},\qquad
 \mathcal F_+ =\{y\in\mathcal F:K(y)>r_*\}.
\]
For the complement map $J(y)=\mathbf1-y$, the total remains $n/2$.
For every root candidate $c$, both child means are complemented, so
\[
 T_c(J(y))=-T_c(y),\qquad
 \operatorname{Gain}(c;J(y))=\operatorname{Gain}(c;y).
\]
The eligible candidates depend only on $X$, and every gain comparison,
positivity decision, and deterministic tie is therefore unchanged.
Consequently $J$ maps $\mathcal F$ bijectively onto itself. Since
\[
 K(J(y))=a-K(y)=2r_*-K(y),
\]
it also maps $\mathcal F_-$ bijectively onto $\mathcal F_+$. Nonemptiness
of $\mathcal F$ implies that both parts are nonempty and
\[
 |\mathcal F_-|=|\mathcal F_+|,\qquad
 \pi_0(\mathcal F_-)=\pi_0(\mathcal F_+)=\tfrac12.
\]
A proposed exchange between a success at index $i$ and a failure at $j$
changes the count by
\[
 K(z)-K(y)=-\1\{i\in A\}+\1\{j\in A\}\in\{-1,0,1\}.
\]
Counts in the two parts differ by at least two, so
$Q(y,z)=0$ for $y\in\mathcal F_-$, $z\in\mathcal F_+$, and also
in the reverse direction. The two parts are closed under $Q$.
By induction, for every $\ell\ge0$,
\[
 Q^\ell(y,\mathcal F_+)=0\quad(y\in\mathcal F_-),\qquad
 Q^\ell(y,\mathcal F_-)=0\quad(y\in\mathcal F_+).
\]
This proves reducibility. Taking the opposite part as $S$ in the
supremum defining total variation proves \eqref{eq:disconnected-tv}.
\end{proof}

Proposition~\ref{prop:finite-mc} remains applicable in this example,
whereas \eqref{eq:disconnected-tv} shows that increasing $L$ cannot
make a chain started at one configuration explore the full conditional
law. Longer runs may improve exploration within a closed part.
Computational effort, count variation, the all-tied fraction, and power
therefore describe different aspects of the calculation. Alternative
reversible proposals retain the same validity proof, while their effects
on power require separate evaluation.

The completed analysis uses equal-tailed ranks of $K$. The exchangeability
proof also permits an upper-tail test of the selected squared score
$(N_RK-aM)^2/(N_Rab)$, or another statistic fixed conditional on
$\mathcal D$ and applied identically to every endpoint. Such orderings
can differ in power and require a prespecified comparison; choosing the
smaller p-value after computing both is not covered by the guarantee.
\section{Simulation study}\label{sec:simulations}
We evaluate three aspects of the procedures developed in
Sections~\ref{sec:homogeneity}--\ref{sec:computation}: rejection under
parent homogeneity, power at a recursively selected third-level target,
and the effect of allocating more computation to chain length or to
the number of reference endpoints. The experiments retain all-cutpoint
Gini selection and compare independent with correlated predictors.

\subsection{Design and evaluation}
Each dataset contains $n\in\{200,400\}$ independent observations with
ten predictors generated as
\begin{equation}
 X_i\overset{\mathrm{iid}}\sim N_{10}(0,\Sigma_\rho),\qquad
 (\Sigma_\rho)_{jk}=\rho^{|j-k|},\qquad \rho\in\{0,0.5\}.
 \label{eq:simulation-design}
\end{equation}
The design is redrawn for every dataset. Conditional on $X$, the outcomes
are independent with $Y_i\sim\operatorname{Bernoulli}\{m(X_i)\}$.
We use the growing rule in Section~\ref{sec:cart-algorithm}, with
maximum depth $D=3$ and minimum child size $n_{\min}=20$. The root
split is level one. The primary test uses $B=199$ reference endpoints,
$L=100$ proposals per chain, and nominal level $\alpha=0.05$,
requiring $(B+1)L=20,000$ proposals per target.

Global-null controls set $m(x)=q$. For independent predictors we use
$q\in\{0.1,0.5\}$; for correlated predictors we use $q=0.3$.
Every existing internal node is tested, with results summarized
separately by level. These controls cover different baseline risks;
the comparison of predictor correlation at a common risk model is
made in the signal experiment below.

To place a target at level three, define generating regions
\begin{align}
 \mathcal R_1&=\{x:x_1\le0\},&
 \mathcal R_2&=\mathcal R_1\cap\{x:x_2\le c_2\},\nonumber\\
 \mathcal A_\star&=\mathcal R_2\cap\{x:x_3\le c_3\},&
 \mathcal B_\star&=\mathcal R_2\setminus\mathcal A_\star .
 \label{eq:simulation-regions}
\end{align}
The thresholds are conditional population medians, chosen so that
$\PP(X\in\mathcal R_2)=1/4$ and
$\PP(X\in\mathcal A_\star)=\PP(X\in\mathcal B_\star)=1/8$.
They are $c_2=c_3=0$ when $\rho=0$ and, to ten decimal places,
$c_2=-0.3935190102$, $c_3=-0.5606328477$ when $\rho=0.5$.
Generate outcomes using
\begin{equation}
 m_\Delta(x)=
 \begin{cases}
  0.3-\Delta/2,&x\in\mathcal A_\star,\\
  0.3+\Delta/2,&x\in\mathcal B_\star,\\
  0.6,&x\in\mathcal R_1\setminus\mathcal R_2,\\
  0.75,&x\notin\mathcal R_1,
 \end{cases}
 \qquad \Delta\in\{0,0.1,0.2,0.3,0.4\}.
 \label{eq:simulation-risk}
\end{equation}
The mean risk in $\mathcal R_2$ remains 0.3. The population mean
risk differences at both generating ancestor splits also remain 0.3,
since the mean in $\mathcal R_1$ is 0.45. Thus varying $\Delta$
changes the target signal while preserving these ancestor contrasts.
The quantiles define the generating regions; CART still searches all
eligible cutpoints. Expected generating-parent sizes are 50 and 100
at $n=200$ and 400, respectively, with expected child sizes 25 and 50.
Realized counts can fall below the minimum-child-size requirement.

For the signal experiment, prespecify the fitted address $v=LL$.
Write $I(C)=\{i:X_i\in C\}$ for a generating region $C$ and define
the target-opportunity event
\begin{equation}
 S=\left\{
 \begin{array}{l}
 H_{LL}(X,Y)\ne\dagger,\quad I_{LL}\subseteq I(\mathcal R_2),\\
 I_{LL}\cap I(\mathcal A_\star)\ne\varnothing,\quad
 I_{LL}\cap I(\mathcal B_\star)\ne\varnothing
 \end{array}\right\}.
 \label{eq:simulation-opportunity}
\end{equation}
This evaluation rule excludes parents containing the fixed outside
signals and requires representation of both generating target children.
It does not require exact recovery of the generating cuts.
On $S$, $\Delta=0$ gives a homogeneous selected parent, whereas
$\Delta>0$ gives a nonhomogeneous parent. The actually fitted children
can themselves have heterogeneous risks, so alternative rejection
measures detection of parent nonhomogeneity.

The known generating regions are used only for evaluating $S$, not
for fitting CART or computing a test. Once $X$ and the selected path
are fixed, $S$ is determined; it adds no restriction to the reference
sampler beyond that conditioning. For each signal setting we report
\begin{equation}
 \PP(S),\qquad \PP(\text{reject}\mid S),\qquad
 \PP(S\cap\{\text{reject}\}).
 \label{eq:simulation-estimands}
\end{equation}
The first measures target availability, the second conditional
rejection, and the third selection followed by rejection among all
generated datasets. Missing or terminal targets contribute zero to
the third quantity and are not counted as performed tests.

Each of the 20 signal settings and six global-null controls has
2,000 dataset replicates. Four independent-feature control settings
are shared with Supplement~\ref{sec:strong-under-weak}; the remaining
22 settings comprise 44,000 new dataset realizations. Replicates are
independent within a setting. In the new study, settings with the same
sample size and replicate index share underlying normal draws and outcome
uniforms. The number of replicates is fixed, rather than the number
satisfying $S$.

We compare the inclusive and randomized procedures in
Section~\ref{sec:mc-procedures}, using the same endpoints.
Randomized rejection is summarized by
\eqref{eq:integrated-randomized-rejection}, integrating only the final
tie uniform. An unadjusted equal-tailed hypergeometric test at the
same fitted split provides a benchmark for the effect of ignoring
selection. For a displayed row, let $W_r$ sum rejection contributions
and $D_r$ count tested targets in dataset $r$. With
$R_{\mathrm{MC}}=2,000$, we estimate
\begin{equation}
 \widehat\pi=\frac{\sum_r W_r}{\sum_r D_r},\qquad
 \widehat{\operatorname{SE}}(\widehat\pi)
 =\frac{\operatorname{sd}(W_r-\widehat\pi D_r)}
        {\sqrt{R_{\mathrm{MC}}}\,\overline D}.
 \label{eq:simulation-mcse}
\end{equation}
This treats a dataset as the independent unit when several nodes
contribute. For conditional signal curves, $D_r=\1\{S_r\}$.
Availability and joint detection use all 2,000 datasets.

\subsection{Rejection under parent homogeneity}
Table~\ref{tab:calibration} shows substantial conservatism for the
inclusive procedure: global-null rejection ranges from 0.26\% to
2.75\%. Randomized rejection ranges from 4.16\% to 5.78\%, compared
with 35.37\% to 99.40\% for the unadjusted test. The latter reflects
the search for a large observed difference, and its alternative
rejection is therefore not used as a power advantage over the
selective tests.

Local-null randomized estimates range from 3.57\% to 5.64\%.
The $n=200$, $\rho=0.5$ local estimate is 3.57\%
(Monte Carlo standard error 0.54 percentage points).
In the correlated global controls, the root estimate at $n=200$ is
4.16\% (0.33 points), and the level-three estimate at $n=400$ is
5.50\% (0.21 points). These three estimates differ from 5\% by more
than 1.96 Monte Carlo standard errors. The table reports finite-run
empirical averages over designs and selected histories, rather than
the exact rejection probabilities for each conditioning event.

The extent of tied ranks depends on the setting. Among independent
global controls with $q=0.1$, all endpoint counts are equal in
38.09\% of level-three tests at $n=200$ and 16.52\% at $n=400$.
For the local-null target, the corresponding percentages across the
four settings are 0--0.31\%. As shown in
Section~\ref{sec:mc-procedures}, an all-tied sample contributes exactly
0.05 to randomized rejection. Calibration and discrimination are
therefore evaluated separately.

\begin{table}[!htbp]
\centering\small
\setlength{\tabcolsep}{3.1pt}
\caption{Rejection under parent homogeneity at nominal 5\%.
Each setting has 2,000 datasets. Entries are rejection percentages
(Monte Carlo standard errors in percentage points).
The tests column counts selected targets; the all-tied column gives
the percentage with identical observed and reference success counts.
Global controls test every existing internal node at the indicated level;
local rows use the event $S$ in \eqref{eq:simulation-opportunity}.}
\begin{tabular}{rrrrrrrrr}
\toprule
$n$ & $\rho$ & $q$ & Level & Tests & Inclusive & Randomized & Unadjusted & All-tied\\
\midrule
\multicolumn{9}{l}{Global homogeneous controls}\\
200 & 0 & 0.1 & 1 & 2,000 & 1.30 (0.25) & 4.95 (0.32) & 91.00 (0.64) & 1.20\\
200 & 0 & 0.1 & 2 & 3,030 & 0.53 (0.13) & 5.00 (0.18) & 61.42 (0.87) & 16.07\\
200 & 0 & 0.1 & 3 & 3,116 & 0.26 (0.09) & 4.91 (0.12) & 35.37 (0.82) & 38.09\\
200 & 0 & 0.5 & 1 & 2,000 & 1.90 (0.31) & 4.72 (0.37) & 98.25 (0.29) & 0.00\\
200 & 0 & 0.5 & 2 & 3,171 & 1.51 (0.22) & 5.42 (0.27) & 82.75 (0.63) & 0.13\\
200 & 0 & 0.5 & 3 & 3,585 & 1.09 (0.17) & 5.14 (0.22) & 66.22 (0.78) & 0.70\\
200 & 0.5 & 0.3 & 1 & 2,000 & 1.30 (0.25) & 4.16 (0.33) & 97.85 (0.32) & 0.00\\
200 & 0.5 & 0.3 & 2 & 3,176 & 0.79 (0.16) & 5.12 (0.24) & 79.38 (0.70) & 0.69\\
200 & 0.5 & 0.3 & 3 & 3,670 & 0.68 (0.14) & 4.98 (0.19) & 62.04 (0.80) & 3.95\\
400 & 0 & 0.1 & 1 & 2,000 & 1.85 (0.30) & 4.90 (0.36) & 96.00 (0.44) & 0.60\\
400 & 0 & 0.1 & 2 & 3,295 & 0.94 (0.17) & 5.35 (0.22) & 81.06 (0.67) & 6.10\\
400 & 0 & 0.1 & 3 & 4,396 & 0.30 (0.08) & 4.79 (0.13) & 62.56 (0.71) & 16.52\\
400 & 0 & 0.5 & 1 & 2,000 & 2.75 (0.37) & 5.78 (0.43) & 99.40 (0.17) & 0.00\\
400 & 0 & 0.5 & 2 & 3,468 & 1.38 (0.20) & 5.11 (0.26) & 90.60 (0.48) & 0.17\\
400 & 0 & 0.5 & 3 & 4,840 & 1.10 (0.15) & 5.08 (0.21) & 79.73 (0.56) & 0.52\\
400 & 0.5 & 0.3 & 1 & 2,000 & 2.75 (0.37) & 5.63 (0.43) & 99.10 (0.21) & 0.00\\
400 & 0.5 & 0.3 & 2 & 3,451 & 1.22 (0.19) & 5.17 (0.24) & 88.35 (0.52) & 0.52\\
400 & 0.5 & 0.3 & 3 & 4,772 & 1.13 (0.15) & 5.50 (0.21) & 77.43 (0.58) & 2.43\\
\midrule\multicolumn{9}{l}{Local homogeneous target ($\Delta=0$, event $S$)}\\
200 & 0 & 0.3 & 3 & 373 & 1.61 (0.65) & 5.64 (0.79) & 46.11 (2.58) & 0.27\\
200 & 0.5 & 0.3 & 3 & 320 & 0.31 (0.31) & 3.57 (0.54) & 38.12 (2.72) & 0.31\\
400 & 0 & 0.3 & 3 & 696 & 1.58 (0.47) & 5.16 (0.60) & 83.62 (1.40) & 0.00\\
400 & 0.5 & 0.3 & 3 & 737 & 1.36 (0.43) & 4.82 (0.56) & 84.67 (1.33) & 0.14\\
\bottomrule
\end{tabular}
\label{tab:calibration}
\end{table}

\FloatBarrier
\subsection{Signal strength and third-level power}
Across the 20 signal settings, 9,552 of 40,000 datasets satisfy $S$;
individual setting counts range from 154 to 745.
Figure~\ref{fig:depth3-power} shows increasing conditional power
as the generating signal becomes large. At $\Delta=0.4$, where the
generating child risks are 0.1 and 0.5, randomized power at $n=200$
is 28.27\% for $\rho=0$ and 15.70\% for $\rho=0.5$.
At $n=400$, these values increase to 70.77\% and 50.97\%.
Inclusive power at $n=400$ is 62.86\% and 39.29\%, respectively.
Thus nontrivial power is available at level three, although it depends
strongly on sample size and signal magnitude.

For smaller differences, sensitivity remains limited. At $n=400$
and $\Delta=0.2$, randomized power is 10.23\% with independent
predictors and 8.25\% with correlated predictors. At $\Delta=0.3$,
these rates are 32.21\% and 24.13\%. In these evaluated settings,
correlation reduces power for the larger signals. The curves average
over changing selected paths and target samples; the one-sided
monotonicity result in Section~\ref{sec:information} concerns a fixed
conditional experiment and does not imply their shape.

Target availability provides a second constraint. At $\Delta=0.4$
and $n=400$, $S$ occurs in 29.35\% of datasets when $\rho=0$ and
21.00\% when $\rho=0.5$. The corresponding probabilities of selection
followed by randomized rejection are 20.77\% and 10.70\%.
At $n=200$, these joint probabilities are 4.23\% and 1.21\%.
A high conditional rejection rate therefore need not translate into
a high chance of detecting the prespecified target from a new dataset.
In this experiment, availability also changes with $\Delta$ despite
the fixed population ancestor contrasts.

\subsection{Effect of the computational budget}
We reuse the $n=400$, $\rho=0.5$ settings with $\Delta=0$ and
$\Delta=0.3$, comparing the primary budget with $B=199,L=500$ and
$B=999,L=100$. Each larger budget uses 100,000 proposals per target,
five times the primary budget. The observed datasets and target
opportunities are identical across budgets; Monte Carlo streams are
independent across budgets. The comparison adds 2,634 finite-run
tests without generating additional datasets.

Figure~\ref{fig:mc-budget} shows little power gain from either
increase. Under $\Delta=0.3$, randomized rejection is 24.13\%,
24.40\%, and 24.54\% for the primary, longer-chain, and more-endpoint
budgets. The paired increases over the primary budget are
0.27 percentage points (Monte Carlo standard error 0.46;
95\% interval $[-0.63,1.17]$) and 0.41 points
(0.34; $[-0.27,1.08]$), respectively.
Inclusive rejection is 16.21\%, 16.03\%, and 15.69\%.
Median native inference times increase from 42.0 milliseconds to
209.8 and 210.2 milliseconds. Under the local null, randomized
rejection is 4.82\%, 5.05\%, and 4.99\%.

All-tied counts occur in 0.14\% of these null targets and in none
of these alternative targets at each budget. Under the alternative,
the mean numbers of distinct counts are 5.08, 5.07, and 5.59.
Thus complete rank ties do not explain the limited power of this
representative calculation. The comparison shows that these two
fivefold budget increases offer little improvement in the evaluated
setting; it does not establish exploration of the full selected fiber
or quantify separately the contributions of conditioning and
computational restrictions.

\begin{figure}[p]
\centering
\includegraphics[width=\textwidth]{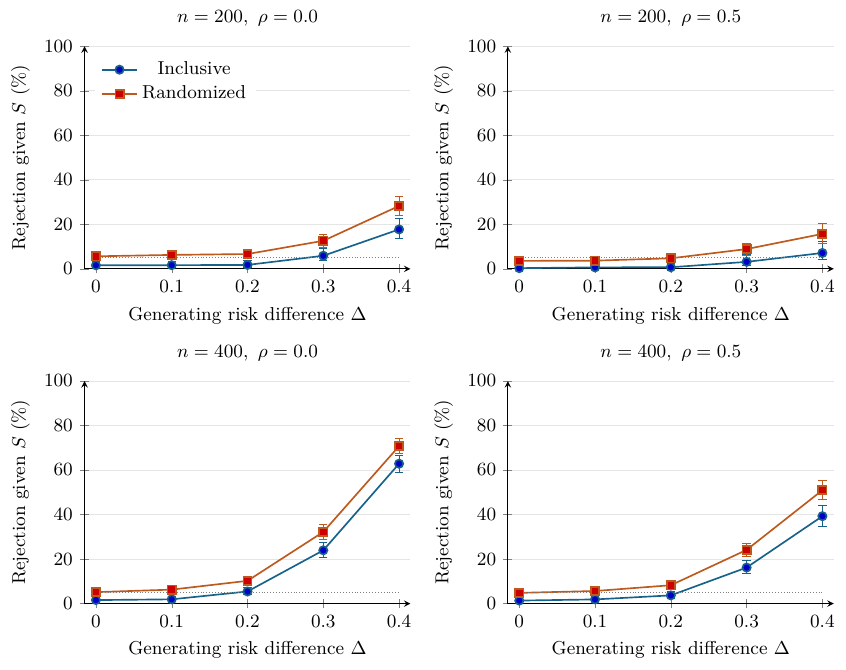}
\par\medskip
{\scriptsize\setlength{\tabcolsep}{2.2pt}
\begin{tabular}{lrrrrr}
\toprule
$n,\rho$ & $\Delta=0$ & $0.1$ & $0.2$ & $0.3$ & $0.4$\\
\midrule
200, 0 & 18.65 / 0.30 / 1.05 & 18.20 / 0.30 / 1.14 & 16.95 / 0.30 / 1.13 & 16.15 / 0.95 / 2.04 & 14.95 / 2.65 / 4.23\\
200, 0.5 & 16.00 / 0.05 / 0.57 & 15.05 / 0.10 / 0.55 & 13.90 / 0.10 / 0.66 & 11.15 / 0.35 / 0.99 & 7.70 / 0.55 / 1.21\\
400, 0 & 34.80 / 0.55 / 1.79 & 36.45 / 0.70 / 2.29 & 35.80 / 1.95 / 3.66 & 33.20 / 7.95 / 10.70 & 29.35 / 18.45 / 20.77\\
400, 0.5 & 36.85 / 0.50 / 1.78 & 37.25 / 0.70 / 2.11 & 35.20 / 1.30 / 2.90 & 29.00 / 4.70 / 7.00 & 21.00 / 8.25 / 10.70\\
\bottomrule
\end{tabular}}
\caption{Power at the third-level target under \eqref{eq:simulation-risk}.
Curves give rejection conditional on $S$. Each setting has 2,000
datasets, with 154--745 target opportunities.
Inclusive error bars are Wilson 95\% intervals; randomized bars are
$\pm1.96$ Monte Carlo standard errors. The dotted line is 5\%.
Each cell below the plots gives availability / joint inclusive
detection / joint randomized detection, all in percent of 2,000.
The curves describe the prespecified target on $S$, not every
selected third-level node.}
\label{fig:depth3-power}
\end{figure}

\begin{figure}[p]
\centering
\includegraphics[width=\textwidth]{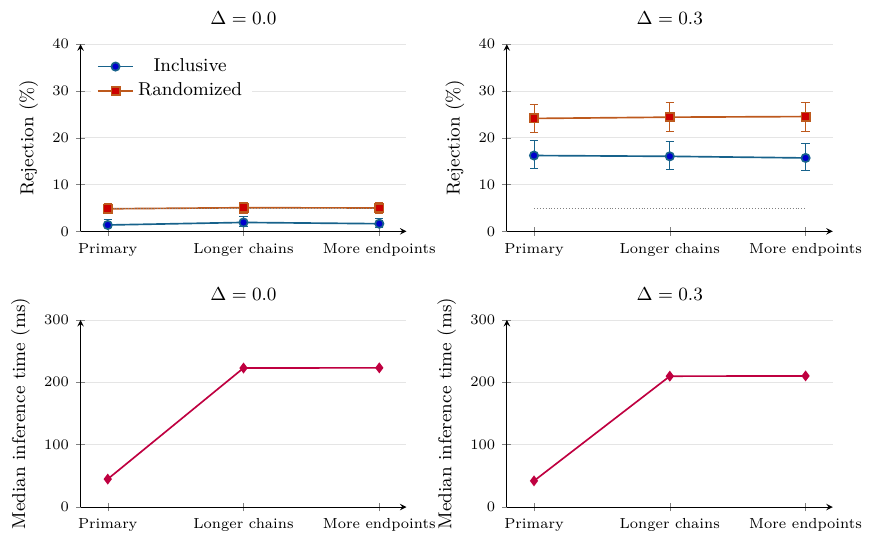}
\par\medskip
{\small\setlength{\tabcolsep}{5pt}
\begin{tabular}{rrrrrrr}
\toprule
$\Delta$ & $B$ & $L$ & Targets & Time (ms) & Distinct $K$ & All-tied (\%)\\
\midrule
0 & 199 & 100 & 737 & 44.8 & 5.07 & 0.14\\
0 & 199 & 500 & 737 & 222.9 & 5.08 & 0.14\\
0 & 999 & 100 & 737 & 223.2 & 5.74 & 0.14\\
0.3 & 199 & 100 & 580 & 42.0 & 5.08 & 0.00\\
0.3 & 199 & 500 & 580 & 209.8 & 5.07 & 0.00\\
0.3 & 999 & 100 & 580 & 210.2 & 5.59 & 0.00\\
\bottomrule
\end{tabular}}
\caption{Computational-budget comparison at $n=400$, $\rho=0.5$.
Primary: $(B,L)=(199,100)$; longer chains: $(199,500)$;
more endpoints: $(999,100)$. The same 737 null and 580 alternative
target opportunities are used at every budget. Top panels show
inclusive and randomized rejection, with intervals as in
Figure~\ref{fig:depth3-power}; bottom panels show median native
inference time. The table also gives the mean number of distinct
counts and the percentage of all-tied samples. Timing excludes tree
fitting and data generation and is measured during four-worker execution.}
\label{fig:mc-budget}
\end{figure}
\FloatBarrier
\section{Discussion}\label{sec:discussion}
Applying established conditional-inference and Monte Carlo principles,
we obtain a finite-sample test for binary Gini CART under parent
homogeneity. The contribution is an explicit conditional count law and
computational procedure accounting for all eligible cutpoints,
deterministic ties, and ancestor selection. The connectivity analysis
and recursive-tree simulations characterize the power and computational
limitations of this CART implementation.

The simulations demonstrate nontrivial power at a third-level target
when the signal is large. With 400 observations and a generating risk
difference of 0.4, conditional randomized power is about 71\% for
independent predictors and 51\% for correlated predictors.
Differences of 0.1--0.2 remain difficult to detect. Moreover, the
probability of reaching the target and rejecting is appreciably lower
than conditional power alone. Even at $n=400$, the generating
third-level parent contains only 100 observations on average.
The full sample size therefore overstates the observations available
for a local comparison.

Including an entire tied block preserves validity but makes the
inclusive procedure conservative. Tie randomization removes this
source of conservatism without adding information to the selected
data. Frequent all-tied calculations in some low-risk controls
illustrate the distinction. In the representative power experiment,
however, all-tied counts were absent, and fivefold budget increases
changed randomized power by less than half a percentage point.
Neither increase offered a substantial gain in that setting.
This does not establish exploration of the full conditional law:
conditional information and connectivity of the swap kernel describe
different restrictions on power.

The scientific interpretation depends on the null. Homogeneity is
a testing target or modeling assumption, not a consequence of fitting
a tree. Under the two-child model, rejection supports unequal
constant child risks. With heterogeneous risks, it need not establish
unequal child averages, and failure to reject does not verify
homogeneity. Such averages can concern the sampled predictor values
or the population within the selected regions.
Supplement~\ref{sec:regional-means} defines these two contrasts and
examines the consequences of interpreting the homogeneity test as
a test of mean equality. When individual risks vary, conditioning on
the parent total leaves weights proportional to
$\exp\{\sum_{i\in R}y_i\operatorname{logit}(p_i)\}$.
The common-factor cancellation then fails, making selective
inference for heterogeneous mean contrasts a distinct problem.

The numerical evidence covers two sample sizes, Gaussian predictor
designs, and the specified signal structure and budgets; it does not
yield a universal minimum node size or establish routine clinical
utility. Pruning, missing predictors, outcome-driven tuning, and an
additional rule for choosing which node to report would require
revisiting the selection event. The guarantee is nodewise rather
than simultaneous over the tree. Applications should state the
scientific role of parent homogeneity or constant child risks and
assess target discovery alongside conditional rejection.

\section*{Acknowledgment}
The authors used GPT-6 Astra by Open AI for English translation of the manuscript. All translated text was checked and revised by the authors, who take full responsibility for the final content.

\clearpage
\appendix
\section{Exploratory assumption-violation experiments}
\label{sec:strong-under-weak}
\subsection{Fixed-design and population mean contrasts}
\label{sec:regional-means}
Parent homogeneity differs from equality of average risks when
individual success probabilities vary within the children. Two average
risk contrasts are relevant. Let $A,B\subseteq\mathbb R^p$ denote
the predictor regions of a specified sibling pair, and write
$I_A=\{i:X_i\in A\}$ and $I_B=\{i:X_i\in B\}$ for their nonempty
sample index sets. Conditional on the observed design, define
\begin{equation}
 \Delta_X(A,B)=\frac{1}{|I_A|}\sum_{i\in I_A}p_i
                  -\frac{1}{|I_B|}\sum_{i\in I_B}p_i.
 \label{eq:design-contrast}
\end{equation}
The fixed-design mean null $\Delta_X(A,B)=0$ permits heterogeneous
probabilities within either child. It is implied by
$H_{0,R}^{\mathrm{hom}}$ with $R=I_A\cup I_B$, but does not imply
that homogeneity condition.

For a population interpretation, suppose observations are iid with
risk function $m(x)=\PP(Y=1\mid X=x)$, so $p_i=m(X_i)$. For regions
of positive probability, the population contrast is
\begin{equation}
 \Delta_P(A,B)=\EE[Y\mid X\in A]-\EE[Y\mid X\in B]
             =\EE[m(X)\mid X\in A]-\EE[m(X)\mid X\in B].
 \label{eq:population-contrast}
\end{equation}
Here the expectations concern a new observation with the fitted
numerical boundaries held fixed. The population mean null
$\Delta_P(A,B)=0$ allows heterogeneity throughout both regions.
It need not give $\Delta_X(A,B)=0$ for the realized design, and
equality of the fixed-design contrast does not establish equality of
the population contrast. Sample-point homogeneity alone also does not
establish population homogeneity. Across training samples, the selected
regions and their contrasts are random.

Under a constant-risk two-child model, the fixed-design contrast is
$p_A-p_B$; the population contrast has the same value when the two
risks are constant throughout the respective predictor regions.
Outside those models, a rejection of parent homogeneity need not imply
a nonzero average-risk contrast. The experiments below examine the
population mean null and record the fixed-design contrast separately.
They assess the consequences of interpreting the homogeneity test as
a test of regional mean equality.

\subsection{Purpose and design}
We next evaluate the homogeneity test when its output is used to infer a
difference in population mean risks. In the heterogeneous cases below the
strong null is false. Rejection is therefore not a type-I error for the
homogeneity hypothesis itself; it is an erroneous mean-difference conclusion
only if that test is interpreted as testing the weak population null.
This distinction is the purpose of the experiment.

We generated $n=200$ or $400$ observations with ten independent standard
normal predictors and grew unpruned Gini trees with all eligible cutpoints,
minimum child size 20, and maximum depth three. The root is level one.
Every selected internal split was tested, conditioning on its full ancestor
path, the observed design, its parent success total, and every outside label.
Numerical thresholds and the fixed tie priority were retained. We extended
the exact integer implementation to arbitrary finite paths, retaining the
same $B=199$, $L=100$ swap-based test as in Section~\ref{sec:computation}.

The primary comparison uses homogeneous risk $m(x)=q$ and a four-way
interaction
\begin{equation}
 m(x)=q+a\prod_{j=1}^{4}\operatorname{sign}(x_j),
 \qquad (q,a)\in\{(0.1,0.08),(0.5,0.4)\}.
 \label{eq:weak-stress-model}
\end{equation}
Individual risks are respectively $0.02/0.18$ or $0.1/0.9$ in the
heterogeneous settings. Every rectangle formed by at most three path cuts
leaves at least one of the four interaction coordinates unconstrained.
Independence and symmetry give mean risk $q$ in every such rectangle, at
every realized numerical boundary. Thus every tested population contrast
through level three is exactly zero, despite within-parent risk variance
$a^2$. This is a deliberately constructed population weak-null example,
not a general model of clinical heterogeneity. The interaction produces no
population mean contrast accessible along these short paths. Its behavior
therefore need not represent more detectable forms of within-parent variation.

Each of the eight primary scenarios uses 2,000 independent training
replicates. A prespecified sensitivity replaces the four-way interaction
by $m(x)=q+a\operatorname{sign}(x_1)\operatorname{sign}(x_2)$, with 1,000
replicates in each of four scenarios. Deeper contrasts can then be nonzero.
For independent normal predictors, integration over a rectangle factorizes:
its risk is $q+a\prod_j\EE[\operatorname{sign}(X_j)\mid l_j<X_j\le u_j]$,
over the interaction coordinates. We use these analytical integrals to
identify truly null selected contrasts and heterogeneous parents. Nonzero
contrasts are retained in the archive and excluded from false-mean-difference
summaries. Known model risks are used only for this evaluation, never for
tree fitting or testing. These scenarios and budgets were fixed before inspecting the
main rejection rates; total training datasets number 20,000.

The weak null here is the population statement
$\Delta_P(\widehat A,\widehat B)=0$. The finite-design contrast
$\Delta_X(\widehat A,\widehat B)$ need not be zero; we record it separately.
Accordingly this is not a demonstration of validity or invalidity under an
exact fixed-design weak null. Homogeneous controls remain valid benchmarks
because their strong-null guarantee holds for every fixed design and hence
after averaging over random designs.

We report inclusive and randomized rejection as defined in
Section~\ref{sec:mc-procedures}. The randomized columns average the
conditional rejection probabilities in
\eqref{eq:integrated-randomized-rejection}, analytically integrating only
the last tie uniform. Rates pool eligible selected nodes within a level;
Monte Carlo standard errors cluster at the independent training-dataset
level. Different scenarios share predictor draws, and paired risk settings
share outcome uniforms; their combined count is not a count of independent
cross-scenario observations. Node counts and the fraction with no variation
in the $B+1$ endpoint success counts accompany the rates. This diagnostic
matters because an all-tied rank produces randomized rejection probability
exactly $0.05$ even when the calculation has no discriminatory information.
The maximum depth does not force every node to split: missing or zero-gain
splits are not silently counted as tests. These are nodewise summaries,
not probabilities of at least one rejection in the whole tree.

\subsection{Results and interpretation}
\begin{table}[tb]
\centering\small
\caption{Strong-null test applied to population weak nulls, maximum depth three.
All primary selected contrasts are exactly zero. Entries are percent rejection
(Monte Carlo standard error in percentage points), from 2,000 training
replicates per setting. Standard errors cluster by training dataset;
randomized entries integrate the final tie uniform.}
\begin{tabular}{rrllrrr}
\toprule
$n$ & $q$ & Risk model & Test & Level 1 & Level 2 & Level 3\\
\midrule
200 & 0.1 & Homogeneous & Inclusive & 1.30 (0.25) & 0.53 (0.13) & 0.26 (0.09)\\
 & &  & Randomized & 4.95 (0.32) & 5.00 (0.18) & 4.91 (0.12)\\
200 & 0.1 & Four-way & Inclusive & 1.05 (0.23) & 0.69 (0.15) & 0.26 (0.09)\\
 & &  & Randomized & 4.70 (0.31) & 5.02 (0.19) & 5.31 (0.15)\\
200 & 0.5 & Homogeneous & Inclusive & 1.90 (0.31) & 1.51 (0.22) & 1.09 (0.17)\\
 & &  & Randomized & 4.72 (0.37) & 5.42 (0.27) & 5.14 (0.22)\\
200 & 0.5 & Four-way & Inclusive & 2.30 (0.34) & 1.22 (0.19) & 0.89 (0.15)\\
 & &  & Randomized & 5.36 (0.41) & 5.17 (0.26) & 4.88 (0.21)\\
400 & 0.1 & Homogeneous & Inclusive & 1.85 (0.30) & 0.94 (0.17) & 0.30 (0.08)\\
 & &  & Randomized & 4.90 (0.36) & 5.35 (0.22) & 4.79 (0.13)\\
400 & 0.1 & Four-way & Inclusive & 1.50 (0.27) & 0.72 (0.15) & 0.45 (0.10)\\
 & &  & Randomized & 4.40 (0.33) & 4.92 (0.20) & 4.91 (0.14)\\
400 & 0.5 & Homogeneous & Inclusive & 2.75 (0.37) & 1.38 (0.20) & 1.10 (0.15)\\
 & &  & Randomized & 5.78 (0.43) & 5.11 (0.26) & 5.08 (0.21)\\
400 & 0.5 & Four-way & Inclusive & 2.70 (0.36) & 1.24 (0.19) & 1.06 (0.15)\\
 & &  & Randomized & 5.80 (0.44) & 4.84 (0.26) & 5.14 (0.20)\\
\bottomrule
\end{tabular}
\label{tab:weak-misuse-primary}
\end{table}
Table~\ref{tab:weak-misuse-primary} reports the completed primary comparison.
Across the four-way heterogeneous settings and all three levels, randomized
rejection ranges from 4.40\% to 5.80\%, and inclusive rejection from
0.26\% to 2.70\%. Matched homogeneous controls give randomized rates
from 4.72\% to 5.78\%. The paired heterogeneous-minus-control estimates
range from -0.50 to 0.65 percentage points. These results show no large
observed inflation across this particular stress-test grid; they do not
establish exact calibration at every setting or selection history.
At $n=200$, $q=0.1$, and level three, the paired increase is 0.40 points,
with an unadjusted 95\% Monte Carlo interval of [0.03,0.77] points.
This is one of several comparisons and is not multiplicity-adjusted;
small departures should not be concealed by describing all rates as nominal.

For the four-way interaction with $q=0.1$, level-three randomized rejection
is 5.31\% at $n=200$ and 4.91\% at $n=400$, compared with inclusive
rates 0.26\% and 0.45\%. There are 3,113 and 4,476 selected
third-level nodes, respectively. Their median parent/smaller-child counts
are 85/24 and 139/30. Endpoint success counts are all equal
in 37.36\% and 17.54\% of those tests. At $q=0.5$, the corresponding
all-equal proportions are 0.54\% and 0.33\%. Thus information
loss and the role of rank randomization differ considerably with prevalence.
The small inclusive rates and, particularly at low prevalence, frequent
all-tied ranks limit what can be inferred from a randomized rate near 5\%.
They are not evidence of useful power against weak-mean alternatives.

\begin{table}[tb]
\centering\small
\caption{Prespecified two-way interaction sensitivity, 1,000 training
replicates per setting. ``All'' counts selected splits; ``Null/heterogeneous''
counts those with zero population mean contrast and a heterogeneous parent.
Rejection percentages (dataset-clustered Monte Carlo standard errors) use
only that latter subset. Nonnull contrasts and homogeneous parents are not
mixed into these rates.}
\begin{tabular}{rrrrrrr}
\toprule
$n$ & $q$ & Level & All & Null/heterogeneous & Inclusive & Randomized\\
\midrule
200 & 0.1 & 1 & 1,000 & 1,000 & 0.90 (0.30) & 4.74 (0.43)\\
200 & 0.1 & 2 & 1,536 & 1,415 & 0.71 (0.22) & 5.40 (0.29)\\
200 & 0.1 & 3 & 1,555 & 1,406 & 0.07 (0.07) & 5.00 (0.17)\\
200 & 0.5 & 1 & 1,000 & 1,000 & 1.20 (0.34) & 3.83 (0.46)\\
200 & 0.5 & 2 & 1,579 & 1,349 & 0.96 (0.27) & 5.24 (0.37)\\
200 & 0.5 & 3 & 1,856 & 1,376 & 0.87 (0.25) & 5.13 (0.34)\\
400 & 0.1 & 1 & 1,000 & 1,000 & 1.40 (0.37) & 4.87 (0.48)\\
400 & 0.1 & 2 & 1,682 & 1,508 & 0.73 (0.22) & 5.04 (0.31)\\
400 & 0.1 & 3 & 2,236 & 1,942 & 0.31 (0.13) & 4.87 (0.20)\\
400 & 0.5 & 1 & 1,000 & 1,000 & 2.80 (0.52) & 5.34 (0.60)\\
400 & 0.5 & 2 & 1,722 & 1,448 & 0.97 (0.26) & 4.62 (0.36)\\
400 & 0.5 & 3 & 2,491 & 1,784 & 0.90 (0.22) & 5.05 (0.33)\\
\bottomrule
\end{tabular}
\label{tab:weak-misuse-sensitivity}
\end{table}
The two-way sensitivity separates genuine deeper mean differences from null
contrasts (Table~\ref{tab:weak-misuse-sensitivity}). For example, with
$q=0.5$, randomized root rejection is 3.83\% at $n=200$ and 5.34\%
at $n=400$. An apparent rejection at a deeper nonnull split is not counted
as a false mean-difference finding. Per-address and selected-root-feature
summaries, paired primary comparisons, and all target values are retained
in the accompanying results archive. These partial stratifications cannot
establish validity conditional on each complete feature/cut history.

Implementation checks independently reproduced 2,560 rational-arithmetic
fits and 313,344 path-membership decisions, including 167,424 at
level three. In 34 comparisons, the complete endpoint vectors matched
the earlier root/second-level implementation with the same random seed.
A separate exhaustive depth-three fiber gave randomized rejection
5.01\% (Monte Carlo standard error 0.03 points) in 6,000 validation
replicates. All 92,720 study tests were subsequently audited by recomputing
both rank-based rejection summaries directly from saved endpoint counts.

The experiment records no large observed rejection inflation in these
particular constructions; it does not provide convincing evidence of
robustness to heterogeneous risks. The high-order interaction and the
limited discrimination of some finite-run calculations weaken that inference.
Its exchangeability proof still requires the strong null. Neither averaging
over designs and selected nodes nor obtaining a rejection rate near 5\%
restores that missing guarantee. This experiment concerns independent
predictors and two interaction constructions; it does not establish validity
for arbitrary heterogeneity, correlated predictors, fixed-design weak nulls,
or useful weak-null power. The strong-null theory and its interpretation
therefore remain as stated in Section~\ref{sec:regional-targets}.

\end{document}